\documentclass[letterpaper, 10 pt, conference]{ieeeconf}  
\usepackage[final]{graphicx}
\usepackage{graphicx}
\usepackage{amsmath}
\usepackage{amsfonts}
\newtheorem{theorem}{Theorem}
\newtheorem{lemma}{Lemma}

\usepackage[switch]{lineno}

\IEEEoverridecommandlockouts                              

\title{\LARGE \bf
Curvature-Guided Safety Filters: State-Dependent Hessian-Weighted Projection with Provable Performance Bounds
}

\author{Ziyan Lin and Liang Xu%
\thanks{Ziyan Lin and Liang Xu are with the Institute of Future Technologies, Shanghai University, 
No.~99 Shangda Road, Baoshan District, Shanghai 200444, China 
(e-mail: ziyan\_lin@shu.edu.cn; liang-xu@shu.edu.cn). 
Liang Xu is the corresponding author.}%
\thanks{This work is supported by the National Natural Science Foundation of China under Grant 62373239, 62333011, 62461160313, and the State Key Laboratory of Space Intelligent Control under grant No. HTKJ2025KL502025.}
}

\begin{document}
\maketitle
\thispagestyle{empty}
\pagestyle{empty}

\begin{abstract}
Safety filters provide a lightweight mechanism for enforcing state and input safety in learning-enabled control. However, common Euclidean projections onto the safe set disregard long-term performance, while directly optimizing the action-value function within the safe set can be nonconvex and computationally prohibitive. This paper proposes a state-dependent, Hessian-guided projection for safety filtering that preserves convexity while improving performance. The key idea is to select a weighted projection matrix from the curvature of the action-value function, thereby biasing the correction toward action directions with higher value sensitivity. We establish (i) a uniform bound on the performance gap between the weighted projection and the safe value-optimal action, and (ii) a condition under which the weighted projection outperforms the Euclidean projection in long-term value. To support black-box controllers, we further present a data-driven construction of the weighted projection matrix via an iterative Q-function learning algorithm with quadratic feature blocks and regularization that enforces curvature dominance and bounded higher-order terms. Simulations on a quadrotor tracking-and-avoidance task indicate that the proposed filter maintains safety while reducing value degradation relative to Euclidean projection, with computational overhead compatible with real-time operation.
\end{abstract}
\section{INTRODUCTION}

With the rapid development of industrial control systems\cite{stouffer2011guide}, autonomous driving\cite{yurtsever2020survey}, and the Internet of Things\cite{fraga2016review}, system safety has become a central concern, directly affecting public safety and the protection of life and property. Traditional control methods mainly focus on stability analysis around set points and reference trajectories\cite{mandal2006introduction}, and often face difficulties in problems with conflicting safety constraints or unstructured environments\cite{osman2010controlling}. Advanced control algorithms such as reinforcement learning (RL) can significantly improve performance but are typically ``black-box'', making their interaction with complex environments hard to predict and potentially risky\cite{10.5555/2789272.2886795}. To address this and jointly consider safety and performance in open, dynamic settings, safety filters (SFs) have recently attracted increasing attention as lightweight, real-time safety mechanisms\cite{wabersich2023data,attar2025data,escudero2025safety}.

The core idea of safety filters is to adjust control commands online so that the system remains in a prescribed safe set while minimally modifying the nominal controller. This allows timely intervention when the system is close to unsafe behavior, preventing constraint violations while retaining as much of the original control performance as possible. Most widely used approaches are projection-based, including Control Barrier Functions (CBFs)\cite{cheng2019end,marvi2021safe,ames2014control,li2024control} and Safety Predictive Filters (SPFs)\cite{wabersich2021predictive}, which explicitly enforce safety by projecting the nominal action onto an admissible set\cite{krasowski2023provablysafereinforcementlearning}. By contrast, reachability-based methods\cite{10068193,10365337,9561561} use Hamilton–Jacobi (HJ) analysis to compute invariant safe sets and backup controllers; although not projection in a strict sense, they can be interpreted as overriding unsafe actions with safe alternatives. In this work, we focus on projection-based safety filters due to their convex structure, computational efficiency, and compatibility with learning-based controllers.

Most existing safety filters enforce safety by minimizing the distance between the nominal action and a certified safe action. However, Euclidean projection can be suboptimal near the safe set boundary where actions must be modified, as conservatism arises mainly from the geometry of the safe set. Several approaches have been proposed to better trade off safety and performance. The CBF-CLF-QP framework enforces safety and stability simultaneously by solving a quadratic program that couples CBFs with Control Lyapunov Functions (CLFs)\cite{choi2020reinforcement,bahreinian2025designing,10907778}, but feasibility can be lost under conflicting constraints or strong model uncertainty. Constrained Markov Decision Processes (CMDPs)\cite{10160991,paternain2022safe,wu2024off} relax safety to expected constraints and optimize long-term performance, but do not guarantee instantaneous safety. Gros et al.\cite{gros2020safe} directly optimize the action-value function within the safe set, achieving a tight safety–performance trade-off, but the resulting nonconvex problem becomes computationally expensive in high-dimensional systems. Designing a safety filter that simultaneously offers strict safety guarantees, good performance, and real-time efficiency thus remains challenging.

{
The main contributions of this work are as follows:
\begin{itemize}
    \item[(1)] We introduce a curvature-guided safety filter that replaces Euclidean projection with a state-dependent Hessian-weighted metric, improving performance while preserving convexity and real-time efficiency.
    
    \item[(2)] We provide performance guarantees under mild regularity assumptions, including a bound on the value loss and a condition under which the Hessian-weighted projection outperforms the Euclidean one.

    \item[(3)] We propose a data-driven scheme to construct the weighting matrices when the action-value function is unavailable. Quadrotor simulations demonstrate improved tracking and lower value degradation with similar computational cost.
\end{itemize}
}

\textbf{Paper Organization:} Section II reviews reinforcement learning, safety filters, and safe Q-learning via projection. Section III introduces the state-dependent weighted projection matrix and establishes performance bounds. Section IV presents a data-driven construction method for black-box controllers. Section V demonstrates the approach on a quadrotor tracking task. Section VI concludes and outlines future directions.

\section{BACKGROUND}
\subsection{Reinforcement Learning}

Reinforcement learning problems are modeled as Markov Decision Processes (MDPs), where at each state $x$ the agent selects an action $u$, receives a reward, and transitions to a new state. The long-term return is summarized by the Bellman value equation
$V(x)=\sum_{u}\pi(u\mid x)\sum_{x'}P_{xx'}^u\!\left[R_{xx'}^u+\gamma V(x')\right],$
where $x'$ denotes the next state, $V(x)$ is the state-value function, $\pi(u \mid x)$ denotes the probability of choosing action $u$ in state $x$, $R_{xx'}^u$ represents the immediate reward, $\gamma$ is the discount factor, and $P_{xx'}^u$ is the transition probability.

Under deterministic policies, the action $u$ is uniquely determined by the state:
$\pi_\star(x)=\arg\max_{u}V(x).$
The Q-function (action-value function) measures the expected cumulative reward when taking action $u$ in state $x$:
\begin{equation}
Q(x, u) = \sum_{x'} P_{xx'}^u \left[ R_{xx'}^u + \gamma V(x') \right]
\label{eq:q_function}
\end{equation}
In Q-learning algorithms, this function is updated iteratively to optimize the control law. The goal is to find action $u_{\text{ref}}$ approximates the optimal action:
\begin{equation}
u_{\text{ref}} = \arg\max_{u} Q_{\boldsymbol{\theta}}(x, u)
\label{eq:q_learning}
\end{equation}
with $Q_\theta$ approximating the optimal $Q_\star$ defined by the Bellman optimality equation
$Q_{\star}(x,u)=\sum_{x'}P_{xx'}^u\!\left[R_{xx'}^u+\gamma\max_{u'}Q_{\star}(x',u')\right].$

\subsection{Safety Filter}

We consider a discrete-time system with safety constraints, where the system evolves at discrete decision steps indexed by $k \in \mathbb{N}$. The dynamics are given by
$
x_{k+1} = f(x_k, u_k),
$
where $x_k \in \mathbb{R}^n$ denotes the state at time step $k$ and $u_k \in \mathbb{R}^m$ denotes the control input. Safety is specified by a continuously differentiable function $h:\mathbb{R}^n \to \mathbb{R}$, with the safe region defined as
$\mathcal{X}_{\mathrm{safe}} := \{\, x \mid h(x) \ge 0 \,\}.$
To ensure that the closed-loop system remains within $\mathcal{X}_{\mathrm{safe}}$, we restrict the control input to a state-dependent admissible set
$u_k \in \mathcal{U}(x_k),\forall k,$
where $\mathcal{U}(x)$ contains all inputs that keep the successor state inside the safe region:
$\mathcal{U}(x)
= \{\, u \in \mathbb{R}^m \mid h(f(x,u)) \ge 0 \,\}.$

If the controller always selects $u_k \in \mathcal{U}(x_k)$, then the safe set $\mathcal{X}_{\mathrm{safe}}$ is forward-invariant, meaning that $x_0 \in \mathcal{X}_{\mathrm{safe}}$ implies $x_k \in \mathcal{X}_{\mathrm{safe}}$ for all future time steps. In practice, the admissible set $\mathcal{U}(x)$ may be computed explicitly using reachability tools or implicitly enforced through control barrier functions or predictive safety filters\cite{miliosic2024stability}. 

Given a known safety set $\mathcal{U}$, we can project nominal actions into the safe region:
\begin{equation}
\begin{aligned}
u_I = \arg\min_{u} \quad & \frac{1}{2} \left\| u - u_{\text{ref}} \right\|^2 \\
\text{s.t.} \quad & u \in \mathcal{U}(x)
\end{aligned}
\label{eq:euclidean_projection}
\end{equation}
This yields the closest safe action in Euclidean norm.

\subsection{Safe Q-learning via Projection}

The “Safe Q-learning via Projection“ method~\cite{gros2020safe} incorporates safety directly into Q-learning by restricting all maximizations of the Q-function to the admissible set $\mathcal{U}(x)$. The safety-constrained optimal action is thus
\begin{equation}
u^\star=\arg\max_{u\in\mathcal{U}(x)}Q(x,u).
\label{eq:safety_constraint_optimization}
\end{equation}
Since both exploration and policy improvement occur within $\mathcal{U}(x)$, every learning update and executed action remains safe. This guarantees forward invariance of $\mathcal{X}_{\mathrm{safe}}$ while optimizing long-term performance over the admissible set.

\section{State-dependent Projection Matrix}
\subsection{Motivation: Beyond Euclidean Projection}
In Safe Reinforcement Learning, once training has converged, the safety filter modifies the reinforcement learning policy by generating an unconstrained optimal action $u_{\text{ref}}$ according to equation~\eqref{eq:q_learning}, and then obtaining the safe action $u_I$ according to equation~\eqref{eq:euclidean_projection}. However, this method only minimizes the Euclidean distance and does not account for the effect on long-term performance. The "Safe Q-learning via Projection" method~\cite{gros2020safe} provides the safe optimal action $u^*$ according to equation~\eqref{eq:safety_constraint_optimization}, using the action-value function as the objective. This can theoretically yield the optimal policy, but the resulting problem is generally nonconvex due to the shape of $Q$ and the structure of $\mathcal{U}(x)$. Solving
~\eqref{eq:safety_constraint_optimization} often requires global search or repeated local optimizations, leading to significant computational overhead and limiting real-time applicability.

\subsection{Hessian-Weighted Projection Filter}
In this paper, we introduce a state-dependent projection matrix \( W(x) \) and modify the optimization objective of the safety filter as follows:

\begin{equation}
\begin{aligned}
u_W = \arg\min_{u} \quad & \| u - u_{\text{ref}} \|_{W(x)}^2 \\
\text{s.t.} \quad & u \in \mathcal{U}(x),
\end{aligned}
\end{equation}
where $u_{\text{ref}}$ is the unconstrained optimal action according to equation~\eqref{eq:q_learning} , \( \| v \|_{W(x)}^2 = v^\top W(x) v \), and the state-dependent projection matrix is chosen as
\begin{equation}
W(x) = -\nabla_u^2 Q(x, u_{\text{ref}}),
\end{equation}
using the Hessian of the action-value function at the reference action to construct the projection matrix.
This choice preserves the convexity of the projection problem while biasing the correction toward action directions that are more critical for long-term performance.

\subsection{Safety Guarantee: Forward Invariance}
We first show that replacing the Euclidean metric with the state-dependent weighting matrix $W(x)$ does not affect the safety guarantees of the original safety filter.

Recall that the safe set is defined as
$\mathcal{X}_{\mathrm{safe}}:=\{x\mid h(x)\ge 0\}$,
and the admissible input set is given by
\[
\mathcal{U}(x):=\{u\in\mathbb{R}^m \mid h(f(x,u))\ge 0\}.
\]
By definition, if the applied control input satisfies $u_k\in\mathcal{U}(x_k)$, then the successor state
$x_{k+1}=f(x_k,u_k)$ satisfies $x_{k+1}\in\mathcal{X}_{\mathrm{safe}}$.

At each time step $k$, the proposed safety filter computes the control input as
\(
u_k := u_W(x_k)=\arg\min_{u}\ \|u-u_{\mathrm{ref}}(x_k)\|^2_{W(x_k)}
\quad\text{s.t.}\quad u\in\mathcal{U}(x_k).
\)
Since $u_W(x_k)$ is a feasible solution of the above constrained optimization problem, it necessarily satisfies
$u_k\in\mathcal{U}(x_k)$.
Therefore,
\[
h(f(x_k,u_k))\ge 0
\quad\Longrightarrow\quad
x_{k+1}\in\mathcal{X}_{\mathrm{safe}}.
\]
Starting from any initial condition $x_0\in\mathcal{X}_{\mathrm{safe}}$, an induction argument shows that
$x_k\in\mathcal{X}_{\mathrm{safe}}$ for all $k\ge 0$.
Hence, the forward invariance of the safe set is preserved.

Importantly, this argument relies only on the constraint set $\mathcal{U}(x)$ and not on the specific choice of the objective function.
Consequently, even when $W(x)$ is state-dependent or learned from data, the proposed Hessian-weighted projection inherits the same safety guarantees as the original safety filter.

\subsection{Performance Rationale: Single-Step Correction with Long-Term Value Sensitivity}
Having established that the weighted projection preserves safety, we now clarify why a single-step correction can still be associated with long-term performance.

In reinforcement learning, the action-value function $Q(x,u)$ represents the expected discounted cumulative return starting from state $x$ and taking action $u$.
As such, $Q$ is inherently a long-term performance metric.
Standard reinforcement learning policies are executed in a single-step fashion, e.g., by selecting
$u_{\mathrm{ref}}(x)=\arg\max_u Q(x,u)$
at each time step, yet they optimize long-term performance precisely because each step is evaluated through the action-value function.

In the presence of safety constraints, the reference action $u_{\mathrm{ref}}$ may be infeasible and must be corrected to lie in the admissible set $\mathcal{U}(x)$.
A Euclidean projection performs this correction purely based on geometric proximity and ignores how deviations in different action directions affect the long-term return.
In contrast, the proposed method incorporates curvature information of $Q$ through the weighting matrix $W(x)$, so that deviations along directions with higher long-term value sensitivity are penalized more strongly.

\subsubsection{Exact Equivalence Under Quadratic $Q$}
The following idealized example formalizes this interpretation by considering a setting in which the action-value function admits an exact quadratic representation.

\begin{lemma}[Exact Long-Term Optimality under Quadratic $Q$]
\label{lem:quadratic_equiv}
Fix a state $x$ and suppose that the action-value function admits an exact strictly concave quadratic form
\[
Q(x,u)=Q(x,u_{\mathrm{ref}})-\tfrac12 (u-u_{\mathrm{ref}})^\top W (u-u_{\mathrm{ref}}),
\]
with $W\succ 0$, so that $u_{\mathrm{ref}}$ is the unconstrained Bellman-optimal action at $x$.
Then the solution of the safe Bellman optimization problem
\[
u^\star = \arg\max_{u\in\mathcal{U}(x)} Q(x,u)
\]
coincides with the output of the Hessian-weighted safety filter
\[
u_W = \arg\min_{u\in\mathcal{U}(x)} \|u-u_{\mathrm{ref}}\|_W^2.
\]
Hence, in this quadratic setting, the proposed single-step safety filter exactly recovers the long-term optimal safe action.
\end{lemma}

\begin{proof}
Fix a state $x$. By assumption, $Q(x,u)$ is the exact Bellman action-value function and is strictly concave quadratic with respect to $u$, so that $u_{\mathrm{ref}}$ is the unique unconstrained maximizer of $Q(x,\cdot)$.
Consequently, the first-order optimality condition
$\nabla_u Q(x,u_{\mathrm{ref}})=0$
holds, and the second-order Taylor expansion of $Q(x,u)$ around $u_{\mathrm{ref}}$ is exact.

It follows that for any $u\in\mathcal{U}(x)$,
\[
Q(x,u)
= Q(x,u_{\mathrm{ref}})
+ \tfrac12 (u-u_{\mathrm{ref}})^\top \nabla_u^2 Q(x,u_{\mathrm{ref}})(u-u_{\mathrm{ref}}).
\]
Since $Q(x,u)$ is strictly concave, $\nabla_u^2 Q(x,u_{\mathrm{ref}})\prec 0$.
Defining the weighting matrix $W:=-\nabla_u^2 Q(x,u_{\mathrm{ref}})\succ 0$, we obtain
\[
Q(x,u)
= Q(x,u_{\mathrm{ref}}) - \tfrac12 \|u-u_{\mathrm{ref}}\|_W^2.
\]

Because $Q(x,u_{\mathrm{ref}})$ is constant with respect to $u$, maximizing the Bellman objective $Q(x,u)$ over the admissible set $\mathcal{U}(x)$ is equivalent to minimizing the weighted deviation $\|u-u_{\mathrm{ref}}\|_W^2$ over $\mathcal{U}(x)$.
Therefore, the solution of the constrained Bellman optimization problem coincides with the Hessian-weighted projection, i.e., $u^\star=u_W$.
\end{proof}

This result shows that when $Q$ is locally quadratic, the proposed weighted projection is not an approximation but yields the exact long-term optimal safe action.
In more general settings, deviations from this ideal case are governed by higher-order terms, whose effects are explicitly controlled in the following analysis through Assumptions~A1--A2.

\subsection{Second-Order Analysis Framework}
To analyze the above method, we introduce the following assumptions:

\begin{itemize}
    \item \textbf{A1. Lipschitz Continuity of the Hessian Matrix} \\
    Suppose that the Hessian matrix of the action value function \( Q(x, u) \in C^2 \) with respect to the action \( u \) is Lipschitz continuous with some Lipschitz constant \( 0 < L_2 < +\infty \), i.e., for any \( u, v \in \mathbb{R}^m \),
    \[
    \|\nabla_u^2 Q(x, u) - \nabla_u^2 Q(x, v)\| \leq L_2 \|u - v\|,
    \]
    where \( \|\cdot\| \) denotes a matrix norm (such as spectral or Frobenius norm).
    
    \item \textbf{A2. Local Negative Definiteness of the Hessian Matrix} \\
    There exists a compact set \( \mathcal{S} \subset \mathbb{R}^m \) such that the weighted action \( u_W \), 
    the Euclidean-projected action \( u_I \), and the safe optimal action \( u^* \) all lie in this set, i.e.,
    \[
    \{u_W, u^*, u_I\} \subseteq \mathcal{S} \overset{\text{def}}{=} \{\, u : \|u - u_{\text{ref}}\|_{W} \le D \,\}.
    \]
    Moreover, there exists \( \mu>0 \) such that \( \nabla_u^2 Q(x,u) \preceq -\mu^2 I \) for all \( u \in \mathcal{S} \),
    meaning that \( Q(x, u) \) is strictly concave with respect to \( u \) within the set \( \mathcal{S} \).
\end{itemize}

\subsubsection{Near-Optimality Bound of the Weighted Projection}

Assumptions A1--A2 are regularity conditions enabling second-order analysis. Note that these are not automatically satisfied by arbitrary learned Q-functions; rather, they are assumptions we impose to enable rigorous analysis and performance bounds. In Section IV, we introduce a data-driven construction method with regularization specifically designed to enforce these conditions, ensuring that the learned Q-function satisfies A1--A2 in practice. Based on these assumptions, we analyze the performance gap using Taylor expansion of the Q-function.

\begin{theorem}
Under Assumptions A1--A2, there exists a constant $C:=L_2 D^3/(3\mu^3)>0$ such that
\[
|Q(x,u_W)-Q(x,u^\ast)| \le C.
\]
\end{theorem}

\begin{proof}
Perform a second-order Taylor expansion of $Q(x,u)$ about $u_{\mathrm{ref}}$ (noting $\nabla_u Q(x,u_{\mathrm{ref}})=0$):
\[
Q(x,u)=Q(x,u_{\mathrm{ref}})+\tfrac{1}{2}(u-u_{\mathrm{ref}})^\top H (u-u_{\mathrm{ref}})+R_3(u),
\]
where $H:=\nabla_u^2 Q(x,u_{\mathrm{ref}})$. 
By Assumption~A1, $\nabla_u^2 Q(x,\cdot)$ is $L_2$-Lipschitz, so the standard Taylor remainder estimate for functions with Lipschitz continuous Hessian (see, e.g.,~\cite{nesterov2004introductory}) yields
\(
|R_3(u)| \le \tfrac{L_2}{6}\|u-u_{\mathrm{ref}}\|^3.
\)
Taking the difference between the two actions gives
$
\begin{aligned}
Q(x,u^*)-Q(x,u_W)
&= \tfrac{1}{2}\Big[(u^*-u_{\mathrm{ref}})^\top H (u^*-u_{\mathrm{ref}})\\
&\quad-(u_W-u_{\mathrm{ref}})^\top H (u_W-u_{\mathrm{ref}})\Big]\\
&\quad + R_3(u^*)-R_3(u_W).
\end{aligned}
$
By construction $W(x)=-H$ and $u_W$ minimizes the quadratic form induced by $W$, hence
\[
(u_W-u_{\mathrm{ref}})^\top W (u_W-u_{\mathrm{ref}})
\le (u^*-u_{\mathrm{ref}})^\top W (u^*-u_{\mathrm{ref}}).
\]
Substituting $W=-H$ yields
\[
(u^*-u_{\mathrm{ref}})^\top H (u^*-u_{\mathrm{ref}})
\le (u_W-u_{\mathrm{ref}})^\top H (u_W-u_{\mathrm{ref}}),
\]
so the quadratic difference is nonpositive. Therefore
\[
\begin{aligned}
|Q(x,u_W)-Q(x,u^*)| &\le |R_3(u^*)-R_3(u_W)| \\
&\le |R_3(u^*)|+|R_3(u_W)|.
\end{aligned}
\]
Using Assumption A2, since $u^*, u_W = \mathcal{S}$ and $\nabla_u^2 Q \preceq -\mu^2 I$ on $\mathcal{S}$, we have $\mu^2 I \preceq W(x) = -\nabla_u^2 Q(x,u_{\mathrm{ref}})$. 
Thus, $\|u-u_{\mathrm{ref}}\|_W^2 \ge \mu^2 \|u-u_{\mathrm{ref}}\|^2$ for any $u$. 
Since $u^*, u_W = \mathcal{S}$ implies $\|u^*-u_{\mathrm{ref}}\|_W, \|u_W-u_{\mathrm{ref}}\|_W \le D$, we obtain
$\|u^*-u_{\mathrm{ref}}\|,\ \|u_W-u_{\mathrm{ref}}\|\le D/\mu$. Thus,
\[
\begin{aligned}
|Q(x,u_W)-Q(x,u^*)|
&\le \tfrac{L_2}{6}\Big(\|u^*-u_{\mathrm{ref}}\|^3+\|u_W-u_{\mathrm{ref}}\|^3\Big)\\
&\le \tfrac{L_2}{6}\cdot 2\left(\frac{D}{\mu}\right)^3
= \frac{L_2 D^3}{3\mu^3}.
\end{aligned}
\]
Setting $C:=L_2 D^3/(3\mu^3)$ completes the proof.
\end{proof}

\subsubsection{When Does Hessian-Weighted Projection Outperform Euclidean Projection}

While the previous theorem establishes that the weighted project action achieves bounded performance loss, a natural question is when the weighted projection actually outperforms the standard Euclidean projection. Our second result gives a sufficient condition: the Hessian-weighted projection outperforms Euclidean projection when the action-value function exhibits strong local curvature (large \(\mu\)) relative to higher-order variations (small \(L_2\)), and when the safety constraint forces significant deviation from \(u_{\mathrm{ref}}\).

\begin{theorem}
    Under the assumptions A1--A2, let \( \delta = \|u_W - u_{\text{ref}}\|_W > 0 \) denote the weighted projection distance, 
    and let \( \beta > 1 \) be the ratio between the Euclidean and weighted projection distances such that
    $\beta \| u_W - u_{\text{ref}} \|_W = \| u_I - u_{\text{ref}} \|_W$
    with 
    \( c = (\beta^2 - 1)/2 \). 
    If the following condition holds:
    \(    \frac{\mu^3}{L_2} \ge \frac{D^3}{3c\,\delta^2},
    \)
    then it is guaranteed that
    \(
    Q(x, u_W) \ge Q(x, u_I).
    \)
\end{theorem}

\begin{proof}
Expanding $Q(x,u)$ around $u_{\mathrm{ref}}$ using the second-order Taylor expansion gives
where, by Assumption~A1 the standard Taylor remainder estimate,
\[
|R_3(u)| \le \tfrac{L_2}{6}\|u - u_{\mathrm{ref}}\|^3, \quad \forall u \in \mathcal{S}.
\]

Define the quadratic approximation
\[
Q_2(x,u) := Q(x,u_{\mathrm{ref}}) + \tfrac{1}{2}(u - u_{\mathrm{ref}})^\top \nabla_u^2 Q(x,u_{\mathrm{ref}})(u - u_{\mathrm{ref}}),
\]
which can be rewritten as
\(
Q_2(x,u) = Q(x,u_{\mathrm{ref}}) - \tfrac{1}{2}\|u - u_{\mathrm{ref}}\|^2_{W(x)},
\)
with $W(x) := -\nabla_u^2 Q(x,u_{\mathrm{ref}}) \succ 0$.

Since $u_W$ is the minimizer in the weighted norm, it follows that
\(
\|u_W - u_{\mathrm{ref}}\|^2_{W(x)} \le \|u_I - u_{\mathrm{ref}}\|^2_{W(x)}.
\)
By Assumption~A2, this relation can be parameterized as
\(
\|u_I - u_{\mathrm{ref}}\|_W = \beta \|u_W - u_{\mathrm{ref}}\|_W, 
\quad \|u_W - u_{\mathrm{ref}}\|_W = \delta,
\)
with $\beta > 1$. Therefore,
\(
Q_2(x,u_W) - Q_2(x,u_I) = \tfrac{1}{2}(\beta^2-1)\|u_W - u_{\mathrm{ref}}\|^2_W = c \, \delta^2,
\)
where $c = (\beta^2-1)/2$.

For the Taylor remainders, Assumption~A2 yields
\(
|R_3(u_I) - R_3(u_W)| \le  \tfrac{L_2}{6}\Big(\|u_I - u_{\mathrm{ref}}\|^3 + \|u_W - u_{\mathrm{ref}}\|^3\Big) \\
\le  \frac{L_2 D^3}{3\mu^3}.
\)

Thus, the Q-value difference can be bounded as
\(
Q(x,u_W) - Q(x,u_I) = \big[ Q_2(x,u_W) - Q_2(x,u_I) \big] \\
 - \big[ R_3(u_W) - R_3(u_I) \big]
 \geq c \, \delta^2 - \frac{L_2 D^3}{3 \mu^3}
\)

and whenever
$c \, \delta^2 \ge \tfrac{L_2 D^3}{3\mu^3}
\Longleftrightarrow
\delta \ge \sqrt{\tfrac{L_2 D^3}{3 c \mu^3}},$
we obtain
$Q(x,u_W) \ge Q(x,u_I).$
\end{proof}

\section{Data-Driven Construction of Weighting Matrices}
In practical systems, the action-value function of the controller is typically unknown. As a result, the projection matrix $W(x)$ cannot be derived analytically and must instead be inferred from data. We therefore rely on an observable dataset of state–action–reward–next-state tuples
\begin{equation}
\mathcal{D} = \{(x_i, u_i, r_i, x_i')\}_{i=1}^N,
\end{equation}
where $N$ is the number of samples, $x_i$ and $u_i$ are the state and action at sample $i$, $r_i$ is the immediate reward, and $x_i'$ is the next state. Our objective is to obtain, in an offline manner, an approximately optimal action-value function $Q_{\theta}(x,u)$ and to extract from its Hessian a state-dependent weighting matrix $W(x)$ that replaces the identity weighting traditionally used in safety filters.

To balance interpretability and computational tractability, we adopt a linear function approximation
\begin{equation}
Q_\theta(x,u) = \theta^\top \psi(x,u),
\end{equation}
where $\psi(x,u)$ is a manually designed, twice-differentiable feature mapping and $\theta$ is the parameter vector to be learned. The mapping $\psi(x,u)$ is decomposed as
\begin{equation}
\psi(x, u) =
\begin{bmatrix}
\psi_c(x) \\
\psi_b(x) u \\
\psi_s(x) \otimes \mathrm{vech}\!\left(\tfrac{1}{2} uu^\top\right) \\
\psi_{\mathrm{nl}}(x,u)
\end{bmatrix},
\quad
\theta =
\begin{bmatrix}
\theta_c \\
\theta_b \\
\theta_s \\
\theta_{\mathrm{nl}}
\end{bmatrix}.
\end{equation}
Here, $\psi_c(x)$ represents state-dependent constant features, $\psi_b(x)u$ encodes the linear dependence on the control input, the quadratic block $\psi_s(x)\otimes\mathrm{vech}(\tfrac{1}{2}uu^\top)$ introduces a structured curvature representation in $u$ that is key for analytic Hessian extraction, and $\psi_{\mathrm{nl}}(x,u)$ denotes a smooth nonlinear residual used to capture higher-order variations. This decomposition allows $Q_\theta(x,u)$ to retain a clear curvature structure while remaining flexible and trainable from data.

Let $\{B_j\}_{j=1}^{m(m+1)/2} \subset \mathbb{S}^m$ (where $m$ is the dimension of the action space) be a symmetric basis such that
$\left[\mathrm{vech}\!\left(\tfrac{1}{2} uu^\top\right)\right]_j = \tfrac{1}{2} u^\top B_j u,$
where $\mathbb{S}^m$ denotes the space of $m \times m$ symmetric matrices, and $\{B_j\}$ is any fixed symmetric matrix basis chosen so that each index $j$ corresponds to one unique quadratic monomial in $u$. Using the fact that
$\nabla_u^2 \!\left(\tfrac{1}{2} u^\top B_j u \right) = B_j,$
the Hessian of $Q_\theta$ with respect to the action is given by
\begin{equation}
\begin{aligned}
\nabla_u^2 Q_\theta(x,u)
=&-\sum_{a=1}^{p_s} \sum_{j=1}^{m(m+1)/2} \theta_{s,aj} \psi_{s,a}(x) B_j \\
&+\sum_{r=1}^{p_{\mathrm{nl}}} \theta_{\mathrm{nl},r} \nabla_u^2 \psi_{\mathrm{nl},r}(x,u),
\end{aligned}
\end{equation}
where $p_s$ and $p_{\mathrm{nl}}$ denote the dimensions of the quadratic and nonlinear feature blocks, respectively, $\psi_{s,a}(x)$ is the $a$-th component of $\psi_s(x)$, and $\theta_{s,aj}$ is the coefficient associated with the pair $(\psi_{s,a}, B_j)$. The first summation collects the curvature contribution of the structured quadratic block, while the second term accounts for the curvature induced by the nonlinear block.

During training, we follow the Approximate Optimal Policy Evaluation paradigm: without relying on a system dynamics model, we reconstruct the $Q$ function from trajectories collected under the nominal controller.  
At iteration $k$, using the current parameter $\theta^{(k)}$, the supervised targets for all samples are computed as:
\begin{equation}
y_i^{(k)} = r_i + \gamma \max_{u'} \big[{\theta^{(k)}}^\top \psi(x_i',u')\big],
\end{equation}
where $\gamma \in (0,1)$ is the discount factor. The maximization can be solved analytically when $u'$ is continuous and $\psi$ has a simple structure, or by enumeration if $u'$ is discrete.  

In the Fitted Q-Iteration least-squares regression step, we introduce regularization terms and update the parameters as:
\begin{equation}
\begin{aligned}
\theta^{(k+1)} = \arg\min_{\theta} &\sum_{i=1}^N \left( y_i^{(k)} - \theta^\top \psi(x_i, u_i) \right)^2\\ &+ \alpha \mathcal{R}_\mu(\theta) + \beta \mathcal{R}_{L_2}(\theta),
\end{aligned}
\end{equation}
where $\alpha, \beta > 0$ are regularization weights. The term $\mathcal{R}_\mu(\theta)$ is constructed from the quadratic basis components to promote the dominance of the second-order terms, with the quadratic coefficient matrix defined as:
\begin{equation}
S_\theta(x)=\sum_{a=1}^{p_s}\sum_{j=1}^{m(m+1)/2}\theta_{s,aj}\psi_{s,a}(x)B_j,
\end{equation}
and computed on a state grid $\mathcal{X}_{\mathrm{grid}}$ as:
\begin{equation}
\mathcal{R}_\mu(\theta)=\sum_{x\in\mathcal{X}_{\mathrm{grid}}}\big\|(-S_\theta(x))_+\big\|_F^2-\tau\log\det(S_\theta(x)+\varepsilon I),
\end{equation}
where $\mathcal{X}_{\mathrm{grid}}$ is a discretized grid of representative states, $(\cdot)_+$ denotes the positive part operator, $\|\cdot\|_F$ is the Frobenius norm, $\tau > 0$ is the logarithmic barrier weight, and $\varepsilon > 0$ is a numerical stability constant. This regularizer penalizes negative eigenvalues and incorporates a logarithmic barrier to maintain a positive spectral lower bound.

The regularization term $\mathcal{R}_{L_2}(\theta)$ controls the higher-order curvature introduced by 
the nonlinear basis functions. 
Since the Hessian Lipschitz condition in Assumption~A1 requires the third-order derivatives of $Q$ 
to be uniformly bounded, we explicitly bound the third-order derivative of each nonlinear feature as:

\begin{equation}
L_{2,r}(x)=\sup_{\|u-u_{\mathrm{ref}}(x)\|\le D}\|\nabla_u^3\psi_{\mathrm{nl},r}(x,u)\|,
\end{equation}
where $D$ is the maximum projection distance from our theoretical analysis, $u_{\mathrm{ref}}(x)$ is the reference control input at state $x$, and $\|\nabla_u^3\psi_{\mathrm{nl},r}(x,u)\|$ denotes the operator norm of the third-order derivative tensor, leading to an estimate of the $L_2$ bound:
\begin{equation}
\widehat L_2(x)=\sum_{r=1}^{p_{\mathrm{nl}}}|\theta_{\mathrm{nl},r}|L_{2,r}(x),
\end{equation}
and the regularization term:
\begin{equation}
\mathcal{R}_{L_2}(\theta)=\sum_{x\in\mathcal{X}_{\mathrm{grid}}}(\widehat L_2(x))^2.
\end{equation}
The combination of these two regularizers ensures that, while approximating the optimal $Q$, the update process is biased toward parameter regimes where the feasibility condition in Theorem~2 is satisfied: $\mathcal{R}_\mu(\theta)$ shapes the quadratic coefficient matrices so that the induced weights $W(x)$ admit a uniformly large spectral lower bound $\mu$, whereas $\mathcal{R}_{L_2}(\theta)$ suppresses the contribution of the nonlinear features to higher-order derivatives and thus keeps the Hessian Lipschitz constant $L_2$ small. Since Theorem~2 guarantees $Q(x,u_W)\ge Q(x,u_I)$ whenever the projection distance $\delta$ satisfies
$\delta\ge\sqrt{{(L_2 D^3)}/{(3 c \mu^3)}},$
increasing $\mu$ and reducing $L_2$ enlarges the set of states and projection directions for which the weighted projection is provably no worse than the Euclidean projection in terms of value.

This process is repeated until $\|\theta^{(k+1)} - \theta^{(k)}\|$ converges below a prescribed threshold. 
once the converged parameter $\hat{\theta}$ is obtained, the state-dependent weighting matrix can be directly computed from the reference control input $u_{\mathrm{ref}}$ as: \(\hat W(x) = -\nabla_u^2 Q_{\hat{\theta}}(x,u_{\mathrm{ref}}).\)

The action-value function $Q$ is learned from data. As a result, the curvature information used to construct the Hessian-weighted projection matrix is perturbed, which in turn affects the filtered action and its true-$Q$ performance. In this subsection, we quantify how the Hessian approximation error propagates to (i) the deviation between the learned weighted-projection action and the ideal weighted projection, and (ii) the corresponding degradation of the performance bounds in Theorems~1--2.

Fix a state $x$ and denote the learned approximation by $\hat Q$. Define
\[
W(x):=-\nabla_u^2 Q(x,u_{\mathrm{ref}}),\qquad
\hat W(x):=-\nabla_u^2 Q_{\hat{\theta}}(x,u_{\mathrm{ref}}),
\]
and let
\[
u_W = \arg\min_{u\in\mathcal U(x)} \tfrac12\|u-u_{\mathrm{ref}}\|_{W(x)}^2,
\]
\[
\hat u_W = \arg\min_{u\in\mathcal U(x)} \tfrac12\|u-u_{\mathrm{ref}}\|_{\hat W(x)}^2 .
\]
Define the Hessian-induced metric error
\begin{equation}
\label{eq:def_DeltaW_noR}
\Delta W(x):=\hat W(x)-W(x).
\end{equation}
Assume that
\(
\|\Delta W(x)\|\le \rho,
0\le \rho<\mu^2,
\)
where $\mu$ is the curvature constant in Assumption~A2. Then
\(
W(x)\succeq \mu^2 I,
\hat W(x)\succeq (\mu^2-\rho)I.
\)
Moreover, as in Assumption~A2, all relevant actions lie in the compact set
\(
\mathcal S:=\{u:\|u-u_{\mathrm{ref}}\|_{W(x)}\le D\},
\)
so in particular $u_W,\hat u_W,u_I,$ and $u^\ast$ belong to $\mathcal S$. Since
$W(x)\succeq \mu^2 I$, we also have
\begin{equation}
\label{eq:euclid_bound_Dmu}
\|u-u_{\mathrm{ref}}\|\le D/\mu,\qquad \forall u\in\mathcal S.
\end{equation}

\begin{lemma}
\label{lem:u_gap_training}
Under the above error model, the learned weighted projection satisfies
\begin{equation}
\label{eq:u_gap_final_noR}
\|\hat u_W-u_W\|
\le
\sqrt{2\rho}\,\frac{D}{\mu^2}.
\end{equation}
\end{lemma}

\begin{proof}
Define the two projection objectives
\[\phi(u):=\tfrac12\|u-u_{\mathrm{ref}}\|_{W(x)}^2,
\qquad
\hat\phi(u):=\tfrac12\|u-u_{\mathrm{ref}}\|_{\hat W(x)}^2 .\]

\emph{Uniform perturbation on $\mathcal S$.}
For any $u\in\mathcal S$, using $\Delta W=\hat W-W$,
\(
|\hat\phi(u)-\phi(u)|
=
\tfrac12\big|(u-u_{\mathrm{ref}})^\top \Delta W (u-u_{\mathrm{ref}})\big|
\le
\tfrac12\|\Delta W\|\,\|u-u_{\mathrm{ref}}\|^2.
\)
By $\|\Delta W\|\le\rho$ and \eqref{eq:euclid_bound_Dmu},
\begin{equation}
\label{eq:obj_perturb_noR}
|\hat\phi(u)-\phi(u)|
\le
\frac{\rho D^2}{2\mu^2},
\qquad \forall u\in\mathcal S .
\end{equation}

\emph{Strong convexity gap for $\phi$.}
Since $W(x)\succeq \mu^2 I$, the function $\phi$ is $\mu^2$-strongly convex. Hence,
for the minimizer $u_W\in\arg\min_{u\in\mathcal U(x)}\phi(u)$,
\begin{equation}
\label{eq:sc_gap_lower_noR}
\phi(u)-\phi(u_W)\ge \tfrac{\mu^2}{2}\|u-u_W\|^2,\qquad \forall u\in\mathcal U(x).
\end{equation}
In particular,
\(
\phi(\hat u_W)-\phi(u_W)\ge \tfrac{\mu^2}{2}\|\hat u_W-u_W\|^2.
\)

\emph{Upper bound via optimality of $\hat u_W$.}
By optimality of $\hat u_W$ for $\hat\phi$,
\(
\hat\phi(\hat u_W)\le \hat\phi(u_W).
\)
Therefore,
\(
\phi(\hat u_W)-\phi(u_W)
\le
\big(\phi(\hat u_W)-\hat\phi(\hat u_W)\big)
+
\big(\hat\phi(u_W)-\phi(u_W)\big).
\)
Applying \eqref{eq:obj_perturb_noR} to $u=\hat u_W$ and $u=u_W$ (both in $\mathcal S$) yields
\(
\phi(\hat u_W)-\phi(u_W)\le \frac{\rho D^2}{\mu^2}.
\)

Combining the last inequality with \eqref{eq:sc_gap_lower_noR} gives
\(
\tfrac{\mu^2}{2}\|\hat u_W-u_W\|^2 \le \frac{\rho D^2}{\mu^2},
\)
which implies \eqref{eq:u_gap_final_noR}.
\end{proof}

We next quantify how the training error in the learned action $\hat u_W$ affects
the near-optimality guarantee of the curvature-guided projection.
Specifically, we upper bound the suboptimality gap with respect to the true
optimal action $u^\ast$ by separating (i) the intrinsic approximation error of
the weighted projection method (captured by the $L_2$--$\mu$ term) and
(ii) the additional degradation caused by using a learned action $\hat u_W$
instead of the ideal optimizer $u_W$ (captured by the gradient bound $G$ and the
training-radius parameter $\rho$ via Lemma~\ref{lem:u_gap_training}).

\begin{theorem}
\label{thm:near_opt_training}
Assume Assumptions~A1--A2 hold for the true $Q$, and define
\(
G:=\sup_{u\in\mathcal S}\|\nabla_u Q(x,u)\|<\infty .
\)
Let $u^\ast=\arg\max_{u\in\mathcal S}Q(x,u)$. Then
\begin{equation}
\label{eq:near_opt_training_noR}
|Q(x,\hat u_W)-Q(x,u^\ast)|
\le
\frac{L_2D^3}{3\mu^3}
+
G\sqrt{2\rho}\,\frac{D}{\mu^2}.
\end{equation}
\end{theorem}

\begin{proof}
By Theorem~1,
\(
|Q(x,u_W)-Q(x,u^\ast)|\le \frac{L_2D^3}{3\mu^3}.
\)
Since $Q(x,\cdot)$ is continuously differentiable on the compact set $\mathcal S$,
the mean value theorem yields
\(
|Q(x,\hat u_W)-Q(x,u_W)|\le G\|\hat u_W-u_W\|.
\)
Combining this with Lemma~\ref{lem:u_gap_training} and applying the triangle inequality
gives \eqref{eq:near_opt_training_noR}.
\end{proof}

We further compare the learned curvature-guided action $\hat u_W$ with the
Euclidean-projection baseline $u_I$.
The following result shows that the performance advantage established in the
ideal (non-training) case degrades by at most the same two error terms as in
\eqref{eq:near_opt_training_noR}. In particular, it yields an explicit
sufficient condition under which the learned solution $\hat u_W$ still
outperforms $u_I$ despite training error.

\begin{theorem}
\label{thm:compare_training}
Under Assumptions~A1--A2, let
\(
\delta=\|u_W-u_{\mathrm{ref}}\|_{W(x)} > 0,
\|u_I-u_{\mathrm{ref}}\|_{W(x)}=\beta\|u_W-u_{\mathrm{ref}}\|_{W(x)},
c=\tfrac{\beta^2-1}{2}.
\)
Then
\begin{equation}
\label{eq:compare_final_noR}
Q(x,\hat u_W)-Q(x,u_I)
\ge
c\,\delta^2
-
\frac{L_2D^3}{3\mu^3}
-
G\sqrt{2\rho}\,\frac{D}{\mu^2}.
\end{equation}
Consequently, a sufficient condition for $Q(x,\hat u_W)\ge Q(x,u_I)$ is
\begin{equation}\label{eq:condition_keep_advantage_noR}
\frac{G\sqrt{2\rho}\,D}{\mu^2}+\frac{L_2D^3}{3\mu^3}\le c\,\delta^2 .
\end{equation}

\end{theorem}

\begin{proof}
Decompose
\(
Q(x,\hat u_W)-Q(x,u_I)
=
\big(Q(x,u_W)-Q(x,u_I)\big)
+
\big(Q(x,\hat u_W)-Q(x,u_W)\big).
\)
By Theorem~2,
\(
Q(x,u_W)-Q(x,u_I)\ge c\,\delta^2-\frac{L_2D^3}{3\mu^3}.
\)
Moreover,
\(
Q(x,\hat u_W)-Q(x,u_W)
\ge
-|Q(x,\hat u_W)-Q(x,u_W)|
\ge
-G\|\hat u_W-u_W\|.
\)
Applying Lemma~\ref{lem:u_gap_training} yields
\(
Q(x,\hat u_W)-Q(x,u_W)
\ge
-\,G\sqrt{2\rho}\,\frac{D}{\mu^2}.
\)
Combining the above inequalities gives \eqref{eq:compare_final_noR}, and
\eqref{eq:condition_keep_advantage_noR} follows by requiring the right-hand side
of \eqref{eq:compare_final_noR} to be nonnegative.
\end{proof}

\noindent\textbf{Discussion.}
Theorems~\ref{thm:near_opt_training}--\ref{thm:compare_training} show that the learning-induced deviation from the ideal optimizer indeed degrades the guarantees through the additional term
$G\sqrt{2\rho}\,D/\mu^2$.
Nevertheless, the structure of the bounds remains consistent with the non-training case: both the near-optimality loss and the advantage over the Euclidean projection are controlled by the same curvature-regularity pair $(L_2,\mu)$, together with the training accuracy parameter $\rho$.
In particular, by choosing a sufficiently large $\mu$ (improving the conditioning of $W(x)$) and operating in regimes with moderate $L_2$ (mild Hessian variation), the sufficient condition \eqref{eq:condition_keep_advantage_noR} can still be satisfied, ensuring that the learned action $\hat u_W$ preserves the performance advantage, i.e., $Q(x,\hat u_W)\ge Q(x,u_I)$, despite training error.

It is important to emphasize that the Q-function learning and weighting matrix construction process is performed entirely offline. During online execution, only a convex projection problem needs to be solved, with computational cost comparable to Euclidean projection.


\section{SIMULATED EXAMPLE}
To evaluate the proposed safety filter, we present simulation results on a nonlinear quadrotor system. 
The system state is defined as
$x=[p^\top,v^\top,q^\top,\omega^\top]^\top \in \mathbb{R}^{13},$
where $p \in \mathbb{R}^3$ denotes position, $v \in \mathbb{R}^3$ linear velocity, $q \in \mathbb{R}^4$ unit quaternion attitude, and $\omega \in \mathbb{R}^3$ angular velocity. The control input is the normalized thrust of four rotors,
$u=[u_1,u_2,u_3,u_4]^\top \in [0,1]^4.$
The nominal controller is a learning-based nonlinear MPC\cite{hewing2020learning}
. The prediction model is formulated as
\begin{equation}
\dot x=f(x)+g(x)u,
\end{equation}
and discretized with a sampling time of $dt=0.1\,\mathrm{s}$ and prediction horizon $N=20$. The cost function is defined as the weighted quadratic sum of tracking error and control effort:
$J=\sum_{i=0}^{N-1}\Big[(x_i-x_{i+1}^{\mathrm{ref}})^\top Q (x_i-x_{i+1}^{\mathrm{ref}})
+u_i^\top R u_i\Big], $
with $Q=\mathrm{diag}(10,10,10,1,1,1,0.1,0.1,0.1,1,1,1,0.1)$ and $R=\mathrm{diag}(0.1,0.1,0.1,0.1)$.
Safety constraints are enforced using zeroing control barrier functions (ZCBFs) of relative degree 2, constructed for each spherical obstacle as
$\ddot h + 2\alpha_1\dot h + \alpha_0^2 h \;\ge\; -\sigma,  \alpha_0=\alpha_1=2.0,$
where $h(p)=\|p-c\|^2-(R+\mathrm{margin})^2$, and $\dot h,\ddot h$ are derived from velocity $v$ and acceleration induced by input $u$. The slack variable $\sigma\ge0$ allows soft constraint satisfaction, and the nominal control $u_{\mathrm{ref}}$ is minimally adjusted only when necessary to ensure safety.

For the weighting matrix, we collected transition samples from diverse nominal MPC trajectories and apply Fitted Q-Iteration with polynomial state features ($p_s=60$) and radial basis functions ($p_{\mathrm{nl}}=20$), using discount factor $\gamma=0.99$ and regularization weights $\alpha=0.5, \beta=0.2$.

Fig.~\ref{fig:quad_trajectory} shows both Euclidean and weighted projection ensure obstacle avoidance, with weighted projection achieving tighter tracking near obstacles. Fig.~\ref{fig:quad_Q} shows improved value preservation, and Fig.~\ref{fig:qual_runtime} demonstrates comparable computational efficiency.
\begin{figure}[htbp]
\centering
\includegraphics[width=\linewidth]{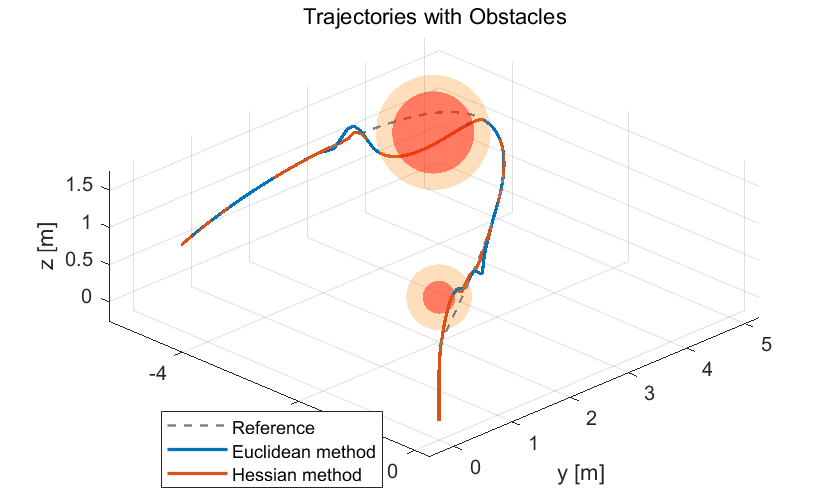}
\caption{Trajectory comparison between the Euclidean projection and the proposed Hessian-weighted projection. 
The dashed curve denotes the reference path, and the shaded circles indicate the obstacle and its safety margin. }
\label{fig:quad_trajectory}
\end{figure}

\begin{figure}[htbp]
\centering
\includegraphics[width=\linewidth]{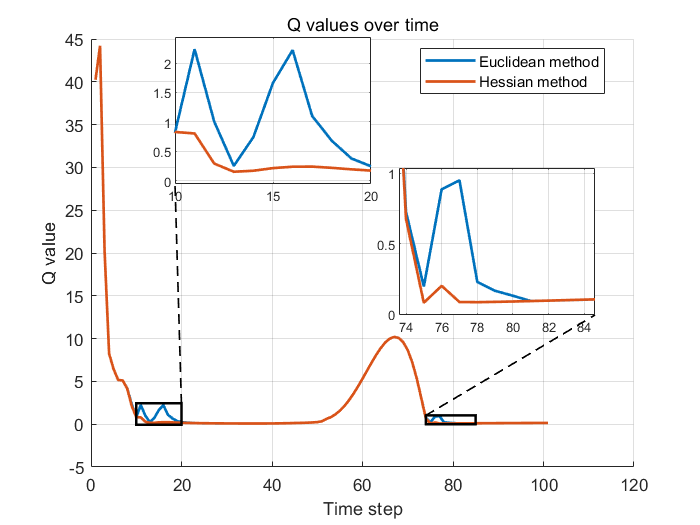}
\caption{Instantaneous difference in action-value function between the Hessian-weighted and Euclidean methods.}
\label{fig:quad_Q}
\end{figure}

\begin{figure}[htbp]
\centering
\includegraphics[width=\linewidth]{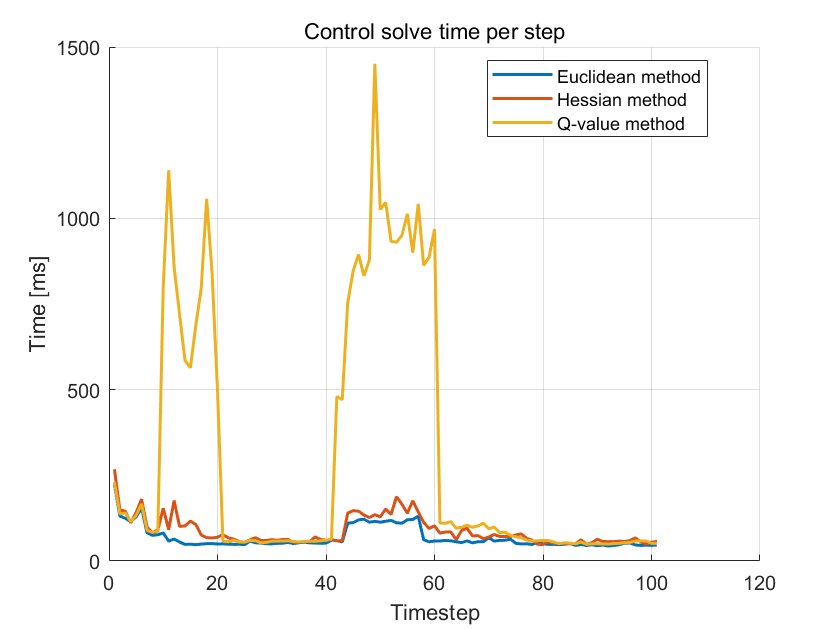}
\caption{IPer-step control solve time for the Euclidean, Hessian-weighted, and Q-value methods.}
\label{fig:qual_runtime}
\end{figure}
\section{CONCLUSION}
We proposed a safety filter that replaces Euclidean projection with a state-dependent weighted projection \(W(x)=-\nabla_u^2 Q(x,u_{\mathrm{ref}})\), using curvature of the action-value function to better preserve long-term performance. Under mild assumptions on Hessian Lipschitzness and negative definiteness, we derived bounds showing that the weighted projection remains near-optimal and can outperform Euclidean projection. To handle black-box controllers, we estimated \(W(x)\) from data via fitted Q-iteration with a structured feature map that enforces dominant quadratic curvature. Quadrotor simulations demonstrate that the method maintains safety while improving tracking and reducing \(Q\)-degradation, all with real-time computational efficiency.

The effectiveness of \(W(x)\) depends on the accuracy of the learned \(Q\)-function around \(u_{\mathrm{ref}}\). Incorporating uncertainty-aware curvature, enabling online adaptation, and improving robustness to model mismatch and time-varying constraints are promising directions.

\begingroup

\bibliographystyle{IEEEtranS}
\bibliography{IEEEabrv, ref}
\endgroup

\end{document}